\documentclass[11pt]{article} %
\usepackage{pgfplots}
\pgfplotsset{width=7cm,compat=1.9}
\usepackage{authblk}
\usepackage{booktabs}
\usepackage{float}
\usepackage[margin=1in]{geometry}
\usepackage[scr=rsfso]{mathalfa}
\usepackage{mdframed} %
\usepackage{microtype} %
\usepackage{mathpazo}
\usepackage{amsthm}
\usepackage{xspace}
\usepackage{comment}
\usepackage{lipsum}
\usepackage{amsfonts}
\usepackage{graphicx}
\usepackage{epstopdf}
\usepackage{algorithm}
\usepackage[noend]{algorithmic}
\usepackage{tikz}
\usepackage{caption}
\usepackage{amsmath}
\usepackage{hyperref}
\usepackage{amssymb}
\usepackage{dirtytalk}
\usetikzlibrary{decorations.pathmorphing}
\ifpdf
  \DeclareGraphicsExtensions{.eps,.pdf,.png,.jpg}
\else
  \DeclareGraphicsExtensions{.eps}
\fi

\newcommand\smallo{
  \mathchoice
    {{\scriptstyle\mathcal{O}}}
    {{\scriptstyle\mathcal{O}}}
    {{\scriptscriptstyle\mathcal{O}}}
    {\scalebox{.7}{$\scriptscriptstyle\mathcal{O}$}}
  }

\DeclareMathOperator*{\argmin}{arg\,min}

\let\epsilon=\varepsilon %

\newcommand{\ST}{\textsc{Undirected Steiner Tree}\xspace}

\newcommand{\SF}{\textsc{Steiner Forest}\xspace}
\newcommand{\DST}{\textsc{Directed Steiner Tree}\xspace}
\newcommand{\MDST}{\textsc{Multi-Rooted Directed Steiner Tree}\xspace}
\newcommand{\setcover}{\textsc{Set Cover}\xspace}
\newcommand{\FL}{\textsc{Facility Location}\xspace}
\newcommand{\GST}{\textsc{Group Steiner Tree}\xspace}
\newcommand{\PG}{\textsc{Projection Games Conjecture}\xspace}

\newcommand{\MIS}{\textsc{Maximum Independent Set}\xspace}
\newcommand{\Pebbling}{\textsc{Pebbling}\xspace}

\def\opt{\mathop{\rm opt}\nolimits}
\def\OPT{\mathop{\rm OPT}\nolimits}
\def\cost{\mathop{\rm cost}\nolimits}
\def\contract{\mathop{\rm contract}\nolimits}

\def\dst{\mathop{\rm DST}\nolimits}

\def\np{\mathop{\rm NP}\nolimits}
\def\p{\mathop{\rm P}\nolimits}
\def\dtime{\mathop{\rm DTIME}\nolimits}
\def\polylog{\mathop{\rm polylog}\nolimits}
\def\poly{\mathop{\rm poly}\nolimits}
\def\zptime{\mathop{\rm ZPTIME}\nolimits}

\newtheorem{theorem}{Theorem} 
\newtheorem{lemma}[theorem]{Lemma} %
\newtheorem{corollary}[theorem]{Corollary} %
\newtheorem{definition}[theorem]{Definition} %
\newenvironment{pproof}[1]{\noindent{\textbf{Proof of #1.}}}{\hfill\rule{2mm}{2mm}}

\begin{document}

\title{An $O(\log k)$-Approximation for Directed Steiner Tree in Planar Graphs}


\author[]{Zachary Friggstad\thanks{Supported by an NSERC Discovery Grant and NSERC Discovery Accelerator Supplement Award.} }
\author[]{Ramin Mousavi}

\affil[]{Department of Computing Science, University of Alberta, Edmonton, Canada.
\authorcr
  \{\tt zacharyf@ualberta.ca, mousavih@ualberta.ca\}}

\date{}

\maketitle

\begin{abstract}
We present an $O(\log k)$-approximation for both the edge-weighted and node-weighted versions of \DST in planar graphs where $k$ is the number of terminals. We extend our approach to \MDST\footnote{In general graphs \MDST and \DST are easily seen to be equivalent but in planar graphs this is not the case necessarily.}, in which we get an $O(R+\log k)$-approximation for planar graphs for where $R$ is the number of roots.
\end{abstract}

\sloppy
\section{Introduction}\label{sec: intro}

In the \DST (DST) problem, we are given a directed graph $G=(V,E)$ with edge costs $c_e\geq 0, e\in E$, a root node $r\in V$, and a collection of terminals $X\subseteq V\setminus\{r\}$. The nodes in $V\setminus(X\cup \{r\})$ are called {\em Steiner} nodes. The goal is to find a minimum cost subset $F\subseteq E$ such that there is an $r-t$ directed path (dipath for short) using only edges in $F$ for every terminal $t\in X$. Note any feasible solution that is inclusion-wise minimal must be an arborescence rooted at $r$, hence the term ``tree''. Throughout, we let $n:=|V|$ and $k:=|X|$.

One key aspect of DST lies in the fact that it generalizes many other important problems, e.g. \setcover, (non-metric, multilevel) \FL, and \GST. Halperin and Krauthgamer \cite{halperin2003polylogarithmic} show \GST cannot be approximated within $O(\log^{2-\epsilon}n)$ for any $\epsilon>0$  unless $\np\subseteq\dtime{(n^{\polylog{(n)}})}$ and therefore the same result holds for DST. 

Building on a height-reduction technique of Calinescu and Zelikovsky \cite{calinescu,zelikovsky1997series}, Charikar et al. give the best approximation for DST which is an $O(k^\epsilon)$-approximation for any constant $\epsilon > 0$ \cite{charikar1999approximation} and also an $O(\log^3 k)$-approximation in $O(n^{{\rm polylog}(k)})$ time (quasi-polynomial time).
This was recently improved by Grandoni, Laekhanukit, and Li \cite{grandoni2019log2}, who give a quasi-polynomial time $O(\frac{\log^2 k}{\log\log k})$-approximation factor for DST. They also provide a matching lower bound in that no asymptotically-better approximation is possible even for quasi-polynomial time algorithms, unless either the \PG fails to hold or $\np \subseteq \zptime(2^{n^\delta})$ for some $0 < \delta < 1$.

The undirected variant of DST (i.e., \ST) is better understood.
A series of papers steadily improved over the simple 2-approximation \cite{zelikovsky199311, karpinski1997new, promel2000new, robins2005tighter} culminating in a $\ln{4}+\epsilon$ for any constant $\epsilon>0$ \cite{byrka2013steiner}.
Bern and Plassmann \cite{bern1989steiner} showed that unless $\p=\np$ there is no approximation factor better than $\frac{96}{95}$ for \ST.

Studying the complexity of network design problems on restricted metrics such as planar graphs and more generally, graphs that exclude a fixed minor has been a fruitful research direction. For example, \cite{borradaile2009n} gives the first {\em polynomial time approximation scheme} (PTAS) for \ST on planar graphs and more generally \cite{bateni2011approximation} obtains a PTAS for \SF on graphs of bounded-genus. Very recently, Cohen-Addad \cite{cohen2022bypassing} presented a {\em quasi-polynomial time approximation scheme} (QPTAS) for Steiner tree on minor-free graphs.

A clear distinction in the complexity of \ST on planar graphs and general graphs have been established; however, prior to our work we did not know if DST on planar graphs is \say{easier} to approximate than in general graphs. Demaine, Hajiaghayi, and Klein \cite{demaine2014node} show that if one takes a standard flow-based
relaxation for DST in planar graphs and further constraints
the flows to be ``non-crossing'', then the solution
can be rounded to a feasible DST solution while losing only a constant factor in the cost. However, the resulting relaxation is non-convex and, to date, we do not know how to compute
a low-cost, non-crossing flow in polynomial time for DST instances on planar graphs. Recently, in \cite{friggstad2021constant} a constant factor approximation for planar DST was given for quasi-bipartite instances (i.e. no two Steiner nodes are connected by an edge). Though, we remark that the techniques in that paper are quite different than the techniques we use in this paper; \cite{friggstad2021constant} uses a primal-dual algorithm based on a standard LP relaxation whereas
the techniques we use in this paper rely on planar separators.

In this paper, we show DST on planar graphs admits an $O(\log k)$-approximation, while DST on general graphs does not have an approximation factor better than $O(\log^{2-\epsilon}n)$ for any $\epsilon>0$ unless $\np\subseteq\dtime{(n^{\polylog{(n)}})}$.

Our approach is based on planar separators presented by Thorup \cite{thorup2004compact}\footnote{As stated in \cite{thorup2004compact} this separator theorem was implicitly proved in \cite{lipton1979separator}.} which states given an undirected graph $G$ with $n$ vertices, one could find a \say{well-structured} subgraph $F$ such that each connected component of $G\setminus F$ has at most $\frac{n}{2}$ vertices. Well-structured separators are useful in enabling divide-and-conquer approach for some problems, such as \MIS and \Pebbling \cite{lipton1980applications}. Also very recently, Cohen-Addad \cite{cohen2022bypassing} uses the same separator we consider to design QPTASes for $k$-MST and \ST on planar graphs. He also develops a new separator to deal with these problems in minor-free graphs. 

We show the separator theorem of Thorup can be used to obtain a simple logarithmic approximation algorithm for planar DST.

\begin{theorem}\label{thm: planar DST}
There is an $O(\log k)$-approximation for planar \DST, where $k$ is the number of terminals.
\end{theorem}

We remark that it is trivial to generalize our algorithm to the node-weighted setting of DST in planar graphs. That is, to instances where Steiner nodes $v \in V\setminus (X \cup \{r\})$ have costs $c_v \geq 0$ and the goal is to find the cheapest $S$ of Steiner Nodes such that the graph $G[\{r\} \cup X \cup S]$ contains an $r-t$ dipath for each $t \in X$. Clearly node-weighted DST generalizes edge-weighted DST even in planar graphs settings since we can subdivide an edge with cost $c_e$ and include this cost on the new node. In general graphs, edge-weighted DST generalizes node-weighted DST because a node $v$ with cost $c_v$ can be turned into two nodes $v^+, v^-$ connected by an edge $(v^+, v^-)$ with cost $c_v$; edges entering $v$ now enter $v^+$ and edges exiting $v$ now exit $v^-$. But this operation does not preserve planarity, it is easy to find examples where this results in a non-planar graph.

We also extend our result to multi-rooted case. In \MDST (MR-DST), instead of one root, we are given multiple roots $r_1,\ldots,r_R$ and the set of terminals $X\subseteq V\setminus\{r_1,\ldots,r_R\}$. The goal here is to find a minimum cost subgraph such that every terminal is reachable from one of the roots.

Note that MR-DST on general graphs is equivalent to DST by adding an auxiliary root node $r$ and adding edges $(r,r_i)$ for $1\leq i\leq R$ with zero cost. However, this reduction also does not preserve planarity. We prove our result for MR-DST by constructing a \say{well-structured} separator for the multi-rooted case.

\begin{theorem}\label{thm: multi-rooted planar DST}
There is an $O(R+\log k)$-approximation for planar \MDST, where $R$ is the number of roots and $k$ is the number of terminals.
\end{theorem}


\section{Preliminaries}\label{sec: prelim}

For convenience, we allow our input graphs to contain multiple directed edges between two nodes.
All 
directed paths (dipath for short) in this paper are simple. Fix a digraph $G=(V,E)$ with edge costs $c_e\geq 0$ for all $e\in E$. We identify a dipath $P$ by its corresponding sequence vertices, i.e., $P=v_1,\ldots,v_a$ and we say $P$ is a $v_1-v_a$-dipath. The {\em start} and {\em end} vertices of $P$ are $v_1$ and $v_a$, respectively. For a subgraph $H$ of $G$, we define the cost of a subgraph $H$ by $\cost_c(H):=\sum\limits_{e\in E(H)}c_e$


We say a vertex $v$ is {\em reachable} from $u$ if there is a dipath from $u$ to $v$. We denote by $d_c(u,v)$ the cost of a shortest dipath from $u$ to $v$, in particular, $d_c(u,u)=0$. The {\em diameter} of a digraph is defined as the maximum $d_c(u,v)$ for all $u\neq v$ where $v$ is reachable from $u$. For both $d_c(.)$ and $\cost_c(.)$ we drop the subscript $c$ if the edge costs is clear from the context. For a subset $S\subseteq V$ and a vertex $u$, we define $d(S,v):=\min\limits_{u\in S}\{d(u,v)\}$. Denote by $G[S]$ the {\em induced subgraph} of $G$ on the subset of vertices $S$, i.e., $G[S]=(S,E[S])$ where $E[S]$ is the set of edges of $G$ with both endpoints in $S$. A {\em weakly connected component} of $G$ is a connected component of the undirected graph obtained from $G$ by ignoring the orientation of the edges. The {\em indegree} of a vertex $v$ with respect to $F\subseteq E$ is the number of edges in $F$ oriented towards $v$. 

A {\em partial arborescence} $T = (V_T, E_T)$ rooted at $r$ in $G$, is a (not necessarily spanning) subgraph of $G$ such that $r \in V_T$ and $T$ is a directed tree oriented away from $r$. An arborescence is a partial arborescence that spans all the vertices. A {\em breadth first search} (BFS) arborescence $B_G$ rooted at $r$ is a (perhaps partial) arborescence including all nodes reachable from $r$ where the dipath from $r$ to any vertex $v$ on $B_G$ is a shortest dipath from $r$ to $v$.

For two disjoint subsets of vertices $S,T\subseteq V$ denote by $\delta(S,T)$ the set of edges with one endpoint in $S$ and the other endpoint in $T$ (regardless of the orientation). 

Given a subgraph $H$ of $G$, for rotational simplicity we write $G/H$ the resulting graph from contracting all the edges in $H$. Also we denote by $G\setminus H$ the resulting graph by removing $H$ from $G$, i.e., removing all the vertices of $H$ and the edges incident to these vertices.


Our algorithm is based on planar separators described by Thorup \cite{thorup2004compact}.
\begin{theorem}[Lemma 2.3 in \cite{thorup2004compact}]\label{thm: separator}
Let $G=(V,E)$ be a connected and undirected planar graph with non-negative vertex weights, and let $T$ be a spanning tree rooted at a vertex $r\in V$. In linear time, one can find three vertices $v_1,v_2$, and $v_3$ such that the union of vertices on paths $P_i$ between $r$ and $v_i$ in $V(T)$ for $i=1,2,3$ forms a separator of $G$, i.e., every connected component of $G\setminus (P_1\cup P_2\cup P_3)$ has at most half the weight of $G$.
\end{theorem}

An immediate consequence of the above result is that given a directed graph and a BFS arborescence rooted at $r$ instead of a spanning tree, one can obtain a separator consisting three shortest dipaths each starting at $r$.

\begin{corollary}[Directed separator]\label{cor: di-separator}
Let $G=(V,E)$ be a planar digraph with edge costs $c_e\geq 0$ for all $e\in E$, and non-negative vertex weights such that every vertex $v \in V$ is reachable from $r$. Given a vertex $r\in V$, in polynomial time, we can find three shortest dipaths $P_1,P_2$, and $P_3$ each starting at $r$ such that every weakly connected component of $G\setminus (P_1\cup P_2\cup P_3)$ has at most half the weight of $G$.
\end{corollary}

Throughout this paper, we create subinstances from $I$ by contracting a subset of edges $F$ in $G$. Whenever, we create a subinstance $I'$ we let the edge cost for the subinstance to be the natural restriction of $c$ to $G/F$, i.e., if $e$ is in both $E(G)$ and $E(G/F)$ then $e$ has cost $c_e$ in $I'$ and if $e$ is in $E(G/F)$ but not in $E(G)$, then its cost in $I'$ is set to be the cost of the corresponding edge in $E(G)$.

Let $I=\big(G=(V,E),c,\{r_1,\ldots,r_R\},X\big)$ be an instance of MR-DST on planar graphs where $G$ is a planar digraph, $c_e\geq 0$ for all $e\in E$ is the edge costs, $\{r_1,\ldots,r_R\}$ are the roots, and $X\subseteq V\setminus\{r_1,\ldots,r_R\}$ is the set of terminals. By losing a small factor in the approximation guarantee, one can assume in an instance of MR-DST that all the costs are positive integers and $d\big(\{r_1,\ldots,r_R\},v\big)$ is polynomially bounded by $n$ for all $v\in V$. The very standard proof
appears in Appendix \ref{app:scale}.

\begin{lemma}[Polynomial bounded distances]\label{lem: poly bounded dist}
For any constant $\epsilon > 0$, if there is an $\alpha$-approximation for MR-DST instances in planar graphs where all edges $e$ have positive integer costs $c_e \geq 1$ and $d_c(r, v) \leq \frac{|X| \cdot |V|}{\epsilon} + |V|$ for each $v \in V$, then
there is an $(\alpha\cdot(1+\epsilon))$-approximation for general instances of MR-DST in planar graphs.
\end{lemma}

\section{Planar DST}
In this section we prove Theorem \ref{thm: planar DST}. Fix an instance $I=\big(G=(V,E),c,r,X\big)$ of DST on planar graphs 
that satisfies the assumptions in Lemma \ref{lem: poly bounded dist} for, say, $\epsilon = 1/2$.
Let $n:=|V|$ and $k:=|X|$. Furthermore, fix an optimal solution $\OPT$ for this instance and let $\opt$ denote its cost. 
So the distance of every vertex from $r$ is at most $O(n \cdot k)$.

Our algorithm recursively constructs smaller subinstances based on a partial arborescence (as a separator) and disjoint subsets of vertices (as the weakly connected components after removing the separator). The following is a more formal definition of these subinstances.

\begin{definition}[Induced subinstances]\label{def: induced subinstances single root}
Let $I=(G=(V,E),c,r,X)$ be an instance of DST on planar graphs. Let $T$ be a partial arborescence rooted at $r$, and let $C_1,\ldots,C_h$ be the weakly connected components of $G\setminus T$. The subinstances of DST induced by tuple $(G,T,C_1,\ldots,C_h)$ are defined as follows: let $G_{\contract}$ be the graph obtained from $G$ by contracting $T$ into $r$. For each $C_i$ where $1\leq i \leq h$ we construct instance $I_{C_i}:=\big(G_{C_i},c,r,C_i\cap X\big)$ where $G_{C_i}:=G_{\contract}[C_i\cup\{r\}]$. See Figure \ref{fig1}.
\end{definition}

Given solutions $\mathcal{F}_1,\mathcal{F}_2,\ldots,\mathcal{F}_h$ for the subinstances induced by $(G,T,C_1,\ldots,C_h)$, one can naturally consider the corresponding subset of edges of $E(T)\cup \mathcal{F}_1\cup \mathcal{F}_2\cup \ldots \cup \mathcal{F}_h$ in $G$ and it is easy to see this forms a feasible solution for instance $I$. We formalize this in the next lemma.

\begin{lemma}[Merged solution]\label{lem: merged soln}
Consider the subinstances $I_{C_i}$ for $1\leq i\leq h$ as defined in Definition \ref{def: induced subinstances single root}. Let $\mathcal{F}_{C_i}$ be a solution for $I_{C_i}$. Let $\mathcal{F}\subseteq E(G)$ be the corresponding edges of $E(T)\cup(\bigcup_{i=1}^h \mathcal{F}_{C_i})$ in $G$. Then, $\mathcal{F}$ is a feasible solution for instance $I$ and furthermore $\cost(\mathcal{F})=\cost(T)+\sum\limits_{i=1}^h \cost(\mathcal{F}_{C_i})$. See Figure \ref{fig1}.
\end{lemma}

\begin{figure}
    \centering
    \includegraphics[scale=0.6]{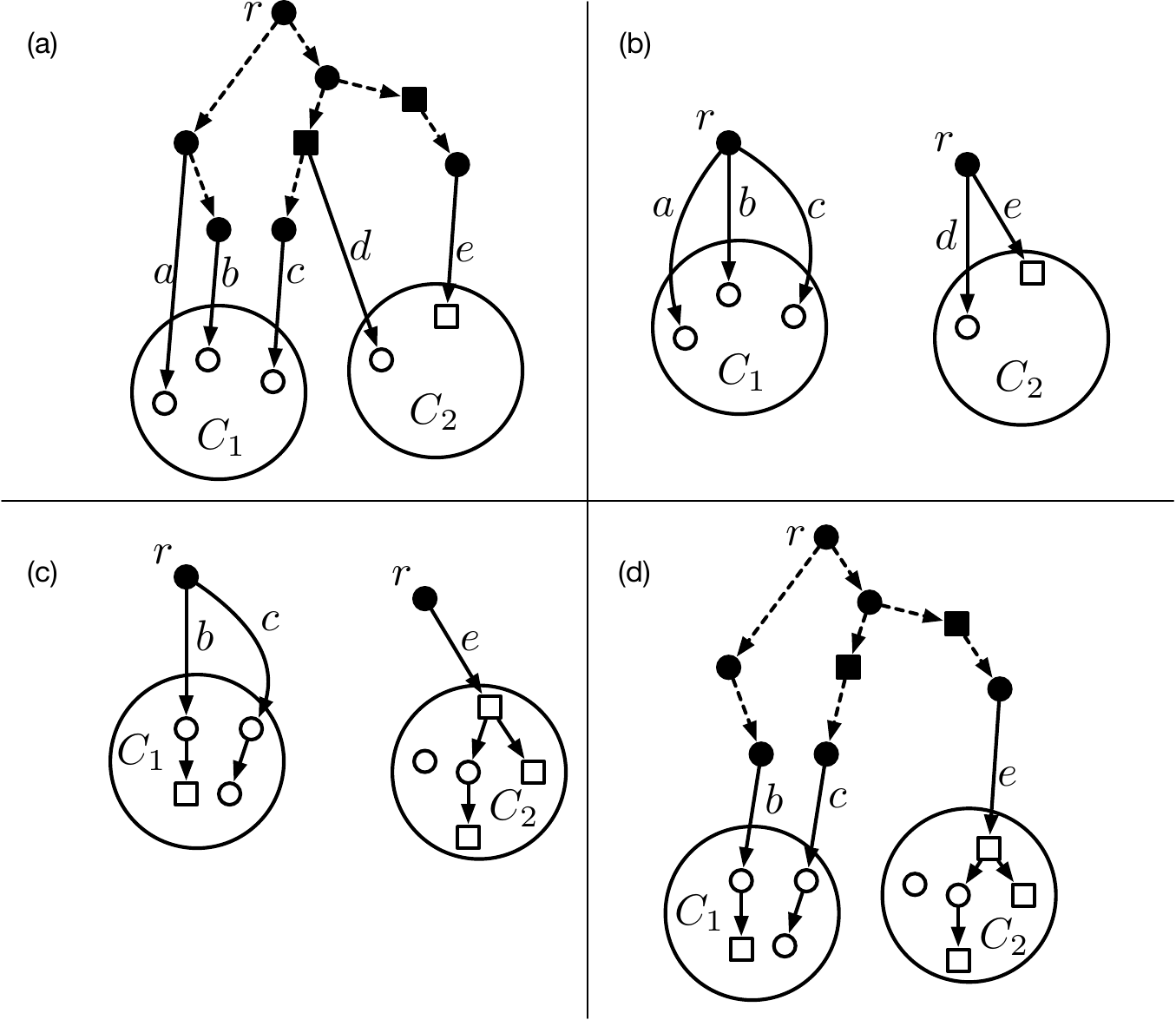}
    \caption{Throughout, squares are terminals and circles are Steiner nodes or the root node $r$. In (a) the separator is shown with dashed edges and solid vertices. The weakly connected components of $G\setminus T$ are shown as circles denoted by $C_1$ and $C_2$. Note that we did not show any edge directed from $C_1$ or $C_2$ into the separator because we can safely remove these edges. In (b) the subinstances $I_{C_1}$ and $I_{C_2}$ induced by $(G,T,C_1,C_2)$ are depicted. In (c), the solutions for each subinstances are shown. Finally, (d) shows how to merge the solutions in (c) to get a solution for the original instance. Note that leaf nodes are not necessarily terminals. One could prune them as a post-processing step, but that is not required by our algorithm.}
    \label{fig1}
\end{figure}

\begin{proof}
The furthermore part is obvious so we prove that $\mathcal{F}$ is feasible for $I$. Consider a terminal node $t\in C_i$. Since $\mathcal{F}_i$ is feasible for $I_{C_i}$, then there is a dipath $P$ from $r$ to $t$. Let $(r,v)$ be the first edge on $P$ and let $(u,v)$ be the corresponding edge to $(r,v)$ in $E(G)$. Then, we must have $u\in V(T)$ as $\delta(C_i,C_j)=\emptyset$ for all $1\leq i\neq j\leq h$. So we can go from $r$ to $u$ in $T$, then take the edge $(u,v)$ and then go from $v$ to $t$ in $\mathcal{F}_{C_i}$. Since all these edges are present in $\mathcal{F}$ and $t$ is an arbitrary terminal, $\mathcal{F}$ is a feasible solution for $I$.
\end{proof}

We first present a high-level idea of a simple $O(\log k)$-approximation that runs in quasi-polynomial time and then with a little extra work, we can make it run in polynomial time with small loss in the approximation guarantee.

\subsection{Warm-up: An overview of a quasi-polynomial time approximation}\label{sec: quasi_polytime}
The algorithm is simple. Fix an optimal solution $\OPT$ with cost $\opt$. First we guess $\opt$. Note by Lemma \ref{lem: poly bounded dist}, $\opt$ is polynomial in $n$ and integral so there are polynomial guesses. Then, we remove all vertices such that their distance from $r$ is more than our guessed value (this is the preprocessing step). For the purpose of separating into subinstances with balanced weight, we let the weight of each terminal to be $1$ and the rest of vertices have zero weight. Apply Corollary \ref{cor: di-separator} and let $P_1,P_2$, and $P_3$ be the resulting shortest dipaths each starting at $r$. Note that $\cost(P_i)\leq \opt$ for $i=1,2,3$ because of the preprocessing step. Let $T:=P_1\cup P_2\cup P_3$, then $T$ is a branching rooted at $r$. Let $C_i$ for $1\leq i\leq h$ be the weakly connected components of $G\setminus T$. Then, we recursively solve the subinstances induced by $(G,T,C_1,\ldots,C_h)$ (see Definition \ref{def: induced subinstances single root}), and finally return the corresponding solution of $E(T)\cup\bigcup_{i=1}^h \mathcal{F}_{C_i}$ in $G$. When the number of terminals in a subinstance becomes one, we can solve the problem exactly by finding the shortest dipath between the root and the only terminal.

Note that each recursive call reduces the number of terminals by half. The guess work for each instance is polynomial in $n$. So it is easy to see the total number of recursive calls is bounded by $n^{O(\log k)}$. Since each time we apply the separator result on an instance $I$, we buy a branching (union of up to three dipaths) of cost at most $3\cdot\opt$, and since the {\em total} cost of optimal solutions across all of the resulting subinstances $I_{C_i}$ is at most $\opt$, a simple induction on the number of terminals shows the final cost is within $(3\cdot\log k+1)\cdot\opt$. A slight improvement to the running time could be made by guessing $\OPT$ within a constant factor (thus only making $O(\log n)$ guesses since all distances are integers bounded by a polynomial in $n$), but the size of the recursion tree would still be $O(\log n)^{O(\log k)}$ which is still not quite polynomial.

\subsection{The polynomial-time algorithm}

The idea here is similar to the quasi-polynomial time algorithm; however, instead of guessing the diameter of an optimal arborescence for each instance, we keep an estimate of it.
Our recursive algorithm tries two different recursive calls: (1) divide the current estimate by half and recurse, or (2) buy a separator and divide the instance into smaller instances and recurs on these instances using the current estimate as the current estimate passed to each smaller instance.

The rationale behind this idea is that if the estimate is close to the optimal value, then our separator is \say{cheap} compared to optimal value so (2) is a \say{good progress} otherwise we make the estimate smaller so (1) is a \say{good progress}. 
The key idea here that leads to polynomial time is that we do not ``reset'' our guess for the optimal solution cost in each recursive call since we know that if our guess is correct for the current instance, then it is an upper bound for the optimal solution cost in each subinstance.

As we mentioned at the beginning, the algorithm is recursive. The input to the algorithm is a tuple $(I,\widetilde{\opt})$ where $\widetilde{\opt}$ is an estimate of $\opt$. The algorithm computes two solutions and take the better of the two. One solution is by a recursive call to $(I,\frac{\widetilde{\opt}}{2})$ and the other one is obtained by applying Corollary \ref{cor: di-separator} to get smaller subinstances and solve each subinstance recursively and merge the solutions as described in Lemma \ref{lem: merged soln}. See Algorithm \ref{alg: log k-approx polytime} for the pseudocode.

\begin{algorithm}
\caption{$\dst(I,\widetilde{\opt})$}
{\bf Input}: $I:=\big(G=(V,E),c,r,X\big)$ and an estimate $\widetilde{\opt}$.\\
{\bf Output}: A feasible solution for instance $I$ or output infeasible.
\begin{algorithmic}
\IF{$\widetilde{\opt}<1$ or $d(r,t) > \widetilde{\opt}$ for some terminal $t \in X$}
\RETURN \texttt{infeasible}
\ELSIF{$|X_I|=1$}
\STATE Let $\mathcal{F}$ be the shortest dipath from $r$ to the only terminal in $X_I$.
\ELSE
\STATE $\mathcal{F}_1 \gets \dst(I,\frac{\widetilde{\opt}}{2})$, if $\mathcal{F}_1$ is \texttt{infeasible} solution then set $\cost(F_1)\gets \infty$. 
\STATE Remove all vertices $v$ with $d(r,v)>\widetilde{\opt}$. \COMMENT{This is the preprocessing step.}
\STATE Apply Corollary \ref{cor: di-separator} to obtain a partial arborescence $T$ consists of up to $3$ shortest dipaths starting at $r$. Let $C_1,\ldots,C_h$ be the weakly connected components of $G\setminus T$. Let $I_{C_i}$ be the $i$-th subinstance induced by $(G,T,C_1,\ldots,C_h)$ for $i=1,\ldots, h$.
\FOR{$i = 1,\ldots , h$}
\STATE $\mathcal{F}'_i\gets \dst(I_{C_i},\widetilde{\opt})$
\ENDFOR
\STATE $\mathcal{F}_2 \gets E(T)\cup(\bigcup\limits_{i=1}^h \mathcal{F}'_i)$, if any $\mathcal F'_i$ is \texttt{infeasible} then set $\cost(\mathcal F_2) \gets \infty$. 
\IF{both $\cost(\mathcal F_1)$ and $\cost(\mathcal F_2)$ are $\infty$}
\RETURN \texttt{infeasible}
\ENDIF
\STATE $\mathcal{F}\gets\argmin\{\cost(\mathcal F_1), \cost(\mathcal F_2)\}$
\ENDIF
\RETURN $\mathcal{F}$.
\end{algorithmic}
\label{alg: log k-approx polytime}
\end{algorithm}
By Lemma \ref{lem: poly bounded dist}, we can assume the edge costs are positive integers and hence $\opt\geq 1$. So if $\widetilde{\opt}<1$, then the output of $\dst(I,\widetilde{\opt})$ is infeasible. The algorithm will terminate since each recursive call either halves $\widetilde{\opt}$ or halves the number of terminals.


\subsection{Analysis}
In this section, we analyze the cost and the running time of Algorithm \ref{alg: log k-approx polytime}. 

\begin{lemma}[Cost and running time]\label{lem: induction cost}
Consider an instance $I=\big(G=(V,E),c,r,X\big)$ and a pair $(I,\widetilde{\opt})$. Let $\ell$ and $\smallo$ be non-negative integers such that $|X|\leq 2^{\ell}$ and $\widetilde{\opt}\leq 2^{\smallo}$. If $\widetilde{\opt}\geq \opt$ where $\opt$ is the optimal value of $I$, then $\dst(I,\widetilde{\opt})$ returns a solution with cost at most $(6\cdot \ell+1)\cdot\opt$. Furthermore, the total number of recursive calls made by $\dst(I,\widetilde{\opt})$ and its subsequent recursive calls is at most $|X|\cdot 2^{2\cdot\ell+\smallo}$.
\end{lemma}
\begin{proof}
First we analyze the cost of the output solution. If $\ell=0$ then we solve $I$ exactly so the statement holds. So for the rest of the proof we assume $\ell\geq 1$. We proceed by induction on $\ell+\smallo\geq 1$.

We assume $\widetilde{\opt}\leq 2\cdot\opt$, otherwise we have $\dst(I,\widetilde{\opt})\leq\dst(I,\frac{\widetilde{\opt}}{2})\leq (6\cdot \ell+1)\cdot\opt$ by induction where the last inequality holds because $\log\frac{\widetilde{\opt}}{2}\leq\log(\widetilde{\opt})-1$. 

Let $\mathcal{F}$ be the solution returned by $\dst(I,\widetilde{\opt})$. Since $\cost(\mathcal{F})\leq\cost(\mathcal{F}_2)$, it suffices to prove $\cost(\mathcal{F}_2)\leq (6\cdot\ell+1)\cdot\opt$. 
Let $\mathcal{F}'_i=\dst(I_{C_i},\widetilde{\opt})$ for $i=1,\ldots, h$ be the solutions constructed recursively for the subinstances. Note that each $I_{C_i}$ for $i=1\ldots,h$ has at most $2^{\ell-1}$ terminals and $\widetilde{\opt}\geq \opt_{I_{C_i}}$ where $\opt_{I_{C_i}}$ is the optimal value of $I_{C_i}$. By the induction hypothesis, we conclude




\begin{equation}\label{eq: cost subinstance}
    \cost(\mathcal{F}'_i)\leq (6\cdot (\ell-1)+1)\cdot\opt_{I_{C_i}}\leq 6\cdot\ell\cdot\opt_{I_{C_i}},~for~i=1,\ldots,h
\end{equation}

Note that $T$ is the union of up to three shortest dipaths and because of the preprocessing step, each shortest dipath starting at $r$ has cost at most $\widetilde{\opt}\leq 2\cdot\opt$. So the following holds:

\begin{equation}\label{eq: cost separator}
    \cost(T)\leq 3\cdot\widetilde{\opt}\leq 6\cdot\opt.
\end{equation}

Combining \eqref{eq: cost subinstance} and \eqref{eq: cost separator} we get:

\begin{align*}
    \cost(\mathcal{F})&=\cost(T)+\sum\limits_{i=1}^h\cost(\mathcal{F}'_i)\\
    &\leq \cost(T) + \sum\limits_{i=1}^h 6\cdot \ell\cdot\opt_{I_{C_i}}\\
    &\leq 6\cdot\opt+6\cdot \ell\cdot \sum\limits_{i=1}^h\opt_{I_{C_i}}\\
    &\leq 6\cdot\opt+6\cdot \ell\cdot\opt\\
    &=(6\cdot \ell+1)\cdot\opt,
\end{align*}
where the first equality follows from Lemma \ref{lem: merged soln}, the first and the second inequalities follow from \eqref{eq: cost subinstance} and \eqref{eq: cost separator}, respectively, and finally the last inequality follows from the fact that $\sum\limits_{i=1}^h\opt_{I_{C_i}}\leq\opt$ as the restriction of $\OPT$ on each $G_{C_i}$ is a feasible solution for $I_{C_i}$ and $G_{C_i}$'s are edge-disjoint.

Next, we analyze the number of recursive calls $R(\ell,\smallo)$ in $\dst(I,\widetilde{\opt})$. We prove by induction on $\ell+\smallo$ that $R(\ell,\smallo)\leq |X|\cdot 2^{2\cdot\ell+\smallo}$. If $\ell=0$, then there is no recursive call. So suppose $\ell\geq 1$. Let $X_i:=|X\cap C_i|\leq \frac{|X|}{2}$ be the number of terminals in subinstance $I_{C_i}$ and let $\ell_i$ be the smallest integer where $|X_i|\leq 2^{\ell_i}$. Since the number of terminals in the subinstances are halved, we have $\ell_i\leq \ell -1$ for all $1\leq i\leq h$. So we can write
\begin{align*}
   R(\ell,\smallo)&=1 + R(\ell,\smallo-1) + \sum\limits_{i=1}^h R(\ell_i,\smallo) \\
   &\leq 1 + |X|\cdot 2^{2\cdot\ell+\smallo-1} + \sum\limits_{i=1}^h |X_i|\cdot 2^{2\cdot\ell_i+\smallo} \\
   & \leq 1 + |X|\cdot 2^{2\cdot\ell+\smallo-1} + 2^{2(\ell-1)+\smallo}\cdot\sum\limits_{i=1}^h |X_i| \\
   & \leq 1 + |X|\cdot 2^{2\cdot\ell+\smallo-1} + 2^{2\cdot\ell+\smallo-2}\cdot |X| \\
   & = 1 + |X|\cdot 2^{2\cdot\ell+\smallo-1} + (2^{2\cdot\ell+\smallo-1}-2^{2\cdot\ell+\smallo-2})\cdot |X|\\
   & = 1 + |X|\cdot 2^{2\cdot\ell+\smallo}-|X|\cdot 2^{2\cdot\ell+\smallo-2}\\
   &\leq |X|\cdot 2^{2\cdot\ell+\smallo},
\end{align*}
where the first inequality follows from the induction hypothesis, the second inequality comes from the fact that $\ell_i\leq \ell-1$, the third inequality holds because $\sum\limits_{i=1}^h |X_i|\leq |X|$, and the last inequality follows from the fact that $|X|\geq 1$ and $\ell\geq 1$.
\end{proof}

\begin{pproof}{Theorem \ref{thm: planar DST}}
For any $\epsilon>0$, we can assume all the shortest dipaths starting at the root are bounded by $\poly(n,\epsilon)$ by losing a $(1+\epsilon)$ multiplicative factor in the approximation guarantee, see Lemma \ref{lem: poly bounded dist}. So we assume properties of Lemma \ref{lem: poly bounded dist} holds for the rest of the proof.

Let $\Delta$ be the maximum distance from the root to any terminal. Let $\widetilde{\opt}:=k\cdot\Delta\leq\poly(n)$. We find a solution by calling $\dst(I,\widetilde{\opt})$.
Applying Lemma \ref{lem: induction cost} with $\widetilde{\opt}:=k\cdot\Delta$, $\ell:=\lceil\log k\rceil\leq\log k+1$ and $\smallo:=\lceil \log \widetilde{\opt}\rceil$ guarantees the solution has cost at most $(6\cdot (\log k+1)+1)\cdot\opt$.

For running time of Algorithm \ref{alg: log k-approx polytime}, we have by Lemma \ref{lem: induction cost} that the number of recursive calls is at most $k\cdot 2^{2\cdot\ell+\smallo}=O(k^4\cdot\Delta)$. So the total number of recursive calls is $\poly(n)$ (recall $k\cdot\Delta=\poly(n)$). The running time within each recursive call is also bounded by $\poly(n)$ so the algorithm runs in polynomial time.
\end{pproof}

\section{Multi-rooted planar DST}

The algorithm for the multi-rooted case is similar to Algorithm \ref{alg: log k-approx polytime}. We need analogous versions of the separator, how we define the subinstances, and how we merge the solutions of smaller subinstances to get a solution for the original instance for the multi-rooted case. 


We start by a generalization of partial arborescence in the single rooted case to multiple roots.

\begin{definition}[Multi-rooted partial arborescence]\label{def: multi-rooted partial arb}
Given a digraph $G=(V,E)$, $R$ vertices $r_1,\ldots,r_R$ designated as roots. We say a subgraph $T$ of $G$ is a multi-rooted partial arborescence if it satisfies the following properties:
\begin{itemize}
    \item[1.] There are vertex-disjoint partial arborescences $T_{i_1},\ldots,T_{i_q}$ rooted at $r_{i_1},\ldots,r_{i_q}$, respectively, and a subset of edges $F\subseteq E\setminus \big(\bigcup_{j=1}^q E(T_{i_j})\big)$, where the endpoints of each edge in $F$ belong to $\bigcup_{j=1}^q V(T_{i_j})$, such that $T=F\cup(\bigcup_{j=1}^q T_{i_q})$.
    \item[2.] $T$ is weakly connected and has no cycle (in the undirected sense).
\end{itemize}
\end{definition}

If a multi-rooted partial arborescence $T$ covers all the vertices in $G$, then we say $T$ is a {\em multi-rooted arborescence} for $G$. See Figure \ref{fig2} for an example.

Fix an instance $I=(G,c,\{r_1,\ldots,r_R\},X)$ of $R$-rooted DST on planar graphs. Next, we present subinstnaces induced by a partial multi-rooted arborescence and bunch of disjoint subsets analogous to Definition \ref{def: induced subinstances single root}.

\begin{definition}[Induced subinstances, multi-rooted]\label{def: induced subinstances mluti-root}
Let $I=(G,c,\{r_1,\ldots,r_R\},X)$ of $R$-rooted DST on planar graphs. Let $T=F\cup(\bigcup_{j=1}^q T_{p_j})$ be a multi-rooted partial arborescence where $T_{p_j}$ is a partial arborescence rooted at $r_{p_j}$ for $1\leq j\leq q$. In addition, let $C_1,\ldots,C_h$ be the weakly connected components of $G\setminus T$. The subinstances of multi-rooted DST induced by tuple $(G,T,C_1,\ldots,C_h)$ are defined as follows: let $G_{\contract}$ be the graph obtained from $G$ by contracting $T$ into a singleton vertex called $r_T$. For each $C_i$ where $1\leq i \leq h$ we construct instance $I_{C_i}:=\Bigg(G_{C_i},c,\{r_T\}\cup\bigg( C_i\cap\Big(\{r_1,\ldots,r_R\}\setminus\{r_{p_1},\ldots,r_{p_q}\}\Big)\bigg),C_i\cap X\Bigg)$ where $G_{C_i}:=G_{\contract}[C_i\cup\{r_T\}]$.
\end{definition}

The following is analogous to Lemma \ref{lem: merged soln} for merging solution in the multi-rooted case.

\begin{lemma}[Merged solutions, multi-rooted]\label{lem: merger soln multi-rooted}
Let $T=F\cup(\bigcup_{j=1}^q T_{p_j})$ be a partial multi-rooted arborescence in $G$. Consider the subinstances $I_{C_i}$ for $1\leq i\leq h$ as defined in Definition \ref{def: induced subinstances mluti-root} and let $\mathcal{F}_{C_i}$ be a solution for $I_{C_i}$. Let $\mathcal{F}\subseteq E(G)$ be the corresponding edges in $(E(T)\setminus F)\cup(\bigcup_{i=1}^h \mathcal{F}_{C_i})$. Then, $\mathcal{F}$ is a feasible solution for instance $I$ and furthermore $\cost(\mathcal{F})=\cost(T\setminus F)+\sum\limits_{i=1}^h \cost(\mathcal{F}_{C_i})$.
\end{lemma}
\begin{proof}
The furthermore part follows directly from the definition of $\mathcal{F}$. We prove $\mathcal{F}$ is feasible for $I$.

Consider a terminal $t$. If $t\in V(T)$, then $t\in V(T_{p_j})$ for some $1\leq j\leq q$ (recall the vertices in $T$ is the union of the vertices in all the partial arborescences $T_{p_j}$'s) so $t$ is reachable from $r_{p_j}$, the root of $T_{p_j}$, in $\mathcal{F}$. Suppose $t\in C_{i}$ for some $1\leq i\leq h$. If $t$ is reachable from a root other than $r_T$ in $\mathcal{F}_{C_{i}}$ then we are done because the same dipath exists in $\mathcal{F}$. So we suppose not and let $P$ be the dipath in $\mathcal{F}_{C_{i}}$ from $r_T$ to $t$. Let $(u,v)$ be the corresponding edge to $(r_T,v)$ in $G$. Note that $u\in V(T_{p_j})$ for some $1\leq j\leq q$ because $\delta(C_s,C_{s'})=\emptyset$ for $1\leq s\neq s'\leq h$. Hence, $t$ is reachable from $r_{p_j}$, the root of $T_{p_j}$, in $\mathcal{F}$ as $E(T_{p_j})\subseteq \mathcal{F}$.
\end{proof}

Given an instance $I$ with roots $r_1,\ldots,r_R$, temporarily add an auxiliary node $r$ and add edges $(r,r_i)$ for all $1\leq i\leq R$ with zero cost (it might destroy the planarity). Run the BFS algorithm as usual rooted at $r$. Then, remove $r$ and all the edges incident to $r$. The result is a vertex-disjoint BFS arborescences $A_1,A_2,\ldots,A_R$ rooted at $r_1,\ldots,r_R$. Note that for every $v\in V(A_i)$, $v$ is closest to $r_i$ than any other roots, i.e., the dipath from $r_i$ to $v$ has cost $d\big(\{r_1,\ldots,r_R\},v\big)$.

Finally, we present the separator result for the multi-rooted case.

\begin{lemma}[A structured separator, multi-rooted]\label{lem: di-separator small numb of subinstances, multi-rooted}
Let $I=(G=(V,E),c,\{r_1,\ldots,r_R\},X)$ be an instance of multi-rooted DST on planar graphs, and let $A_1,\ldots,A_R$ be the vertex-disjoint BFS arborescence rooted at $r_1,\ldots,r_R$. There is a multi-rooted partial arborescence $T=F\cup(\bigcup\limits_{j=1}^R T_{i_j})$, where $T_{i_j}$ could possibly be empty (i.e., with no vertices) such that the following hold:
\begin{itemize}
    \item[(a)] $T_j$ is either empty or is a subtree of $A_j$ rooted at $r_j$ that consists of the union of up to four shortest dipaths each starting at $r_j$.  
    \item[(b)] Let $C_1,\cdots,C_h$ be the weakly connected components of $G\setminus T$. Then, each subinstance $I_{C_i}$ induced by $(G,T,C_1,\ldots,C_h)$ has at most $\frac{|X|}{2}$ terminals for $1\leq i\leq h$.
   \item[(c)] Let $\mathcal{F}_i$ be a solution to subinstance $I_{C_i}$ for $1\leq i \leq h$. Then, the corresponding solution $(E(T)\setminus F)\cup (\bigcup\limits_{i=1}^h\mathcal{F}_i)$ in $G$ is feasible for $I$ with cost exactly $\cost(T\setminus F)+\sum\limits_{i=1}^h\cost(\mathcal{F}_i)$.
\end{itemize}
\end{lemma}

\begin{proof}
Figure \ref{fig2} helps to visualize this proof. 

Since $G$ is weakly connected, there is a subset of edges $F'$ in $G$ such that $T':=F'\cup(\bigcup_{i=1}^R A_i)$ is a multi-rooted arborescence of $G$ (spanning all the vertices) and the endpoints of edges in $F$ are in $\bigcup\limits_{i=1}^R V(A_i)$. Make $T'$ rooted at an arbitrarily chosen root, say $r_1$. Apply Theorem \ref{thm: separator} with terminal vertices having weight $1$ and the rest of vertices having weight $0$, and $T'$ as the spanning tree (in the undirected sense). This gives three paths $P_1,P_2$, and $P_3$ in $T'$ each with starting vertex $r_1$ such that every weakly connected component $C_i$ of $G\setminus(P_1\cup P_2\cup P_3)$ has at most $\frac{|X|}{2}$ terminals for $1\leq i \leq h$. Note, these three paths do not necessarily follow the directions of the edges.

Fix $A_i$ for some $1\leq i\leq R$ and a path $P_j:=(r_1=v_1),v_2,\ldots,v_N$ for $1\leq j\leq 3$. Let $a$ and $b$ (possibly $a=b$) be the smallest and the largest indices, respectively, such that $v_a$ and $v_b$ are in $V(A_i)$. We claim the subpath $P_{[a,b]}:=v_a,v_{a+1},\ldots,v_b$ is a subgraph of $A_i$. Suppose not, so there must be two indices $a\leq a'< b'\leq b$ such that $v_{a'},v_{b'}\in V(A_i)$ and $v_{a'+1},v_{a'+2},\ldots,v_{b'-1}\notin V(A_i)$. Let $P^{a'}_{A_i}$ and $P^{b'}_{A_i}$ be the paths from $r_i$ to $a'$ and $b'$ in $V(A_i)$, respectively. So $P^{a'}_{A_i}\cup P^{b'}_{A_i} \cup P_{[a',b']}$ forms a cycle in $T'$, a contradiction. Furthermore, for $j=1,2,3$ let $v_j$ be the closest vertex to $r_1$ on $P_j$ (in terms of edge hops) that is in $A_i$ as well (if exists). Then, $v_1=v_2=v_3$ as otherwise we have a cycle in $T'$ because all $P_j$'s start at $r_1$.

For each $1\leq i \leq R$ and $1\leq j\leq 3$, we {\bf mark} the nodes with smallest and largest indices in $P_j$ that are in $A_i$. We proved above, that the number of these marked vertices in each $A_i$ is at most $4$. Furthermore, $(P_1\cup P_2\cup P_3)\cap A_i$ is a subgraph of the union of dipaths from $r_i$ to each marked vertices in $A_i$ for all $1\leq i \leq h$.

We construct our partial multi-rooted arborescence $T$ as follows: let $T_i$ be the union of (up to four) shortest dipaths from $r_i$ to the marked vertices in $A_i$. Let $F:=E\big(P_1\cup P_2\cup P_3\big)\setminus (\bigcup_{i=1}^R E(T_i))$ which is the subset of edges whose endpoints are in different $V(A_i)$'s, i.e., $F\subseteq F'$. Let $T:=F\cup(\bigcup_{i=1}^R T_i)$. Note that for $A_i$'s with no marked vertices, $T_i$ is empty (with no vertices not even $r_i$). Since $T$ is a partial multi-rooted arborescence that contains $P_1\cup P_2\cup P_3$ as a subgraph, every weakly connected components of $G\setminus T$ has at most $\frac{|X|}{2}$ terminals. This finishes the proof of parts (a) and (b).

Property (c) follows from Lemma \ref{lem: merger soln multi-rooted} and the fact that the conditions in Lemma \ref{lem: merger soln multi-rooted} are satisfied.
\end{proof}

\begin{figure}
    \centering
    \includegraphics[scale=0.7]{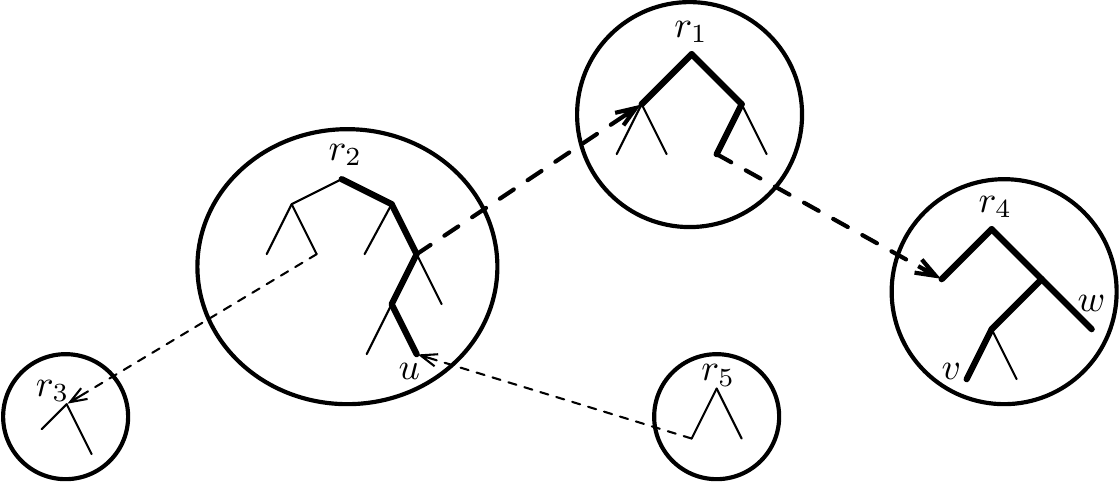}
    \caption{A depiction of the multirooted separator in an instance with $R = 5$ roots. The solid edges (thick and thin) are the shortest-path arborescences $A_i$ for $i = 1, \ldots, R$. The dashed edges are $F'$, they exist solely to allow us to apply Theorem \ref{thm: separator} starting from a spanning tree of the underlying undirected graph and to witness the contraction of all vertices on the thick edges results in a planar graph. After applying Theorem \ref{thm: separator}, we get three vertices depicted as $u,v,w$. The vertices touching the thick and solid edges then form the multirooted separator: these include all vertices lying on paths from $r$ to $u,v,$ or $w$ (as in Theorem \ref{thm: separator}). Additionally, for each $i = 1, \ldots, R$ that includes at least one node from some $r_1-a$ path for some $a \in \{u,v,w\}$, the multirooted separator includes vertices on the unique path connecting $r_i$ to the $r-a$ path (eg. the path from $r_2$ to the $r_1-u$ path).\\
    In the algorithm, the solution will purchase the thick solid edges, but not the thick dashed edges. However, we do contract all thick edges (dashed and solid) to generate the subproblems: the number of roots also drops by 2 since the separator touches 3 shortest-path arborescences. Any solution that is connected from the new contracted root will be connected from either $r_1, r_2$ or $r_4$ using the thick and solid edges after uncontracting.}
    \label{fig2}
\end{figure}

The algorithm for the multi-rooted version is the same as Algorithm \ref{alg: log k-approx polytime} with the following two tweaks: (1) in the preprocessing step we remove vertices $v$ where $d\big(\{r_1,\ldots,r_R\},v\big)>\widetilde{\opt}$, and (2) instead of Corollary \ref{cor: di-separator} we apply Lemma \ref{lem: di-separator small numb of subinstances, multi-rooted} to obtain the subinstances.

Next, we analyze the cost and the running time of this algorithm.

\begin{lemma}[Cost and running time, multi-rooted]\label{lem: induction cost, multi-rooted}
Consider an instance $I=\big(G=(V,E),w,\{r_1,\ldots,r_R\},X\big)$ and a pair $(I,\widetilde{\opt})$. Let $\ell$ and $\smallo$ be non-negative integers such that $|X|\leq 2^{\ell}$ and $\widetilde{\opt}\leq 2^{\smallo}$. If $\widetilde{\opt}\geq \opt$ where $\opt$ is the optimal value of $I$, then $\dst(I,\widetilde{\opt})\leq \big(8\cdot(R+ \ell)+1\big)\cdot\opt$ and the number of recursive calls is at most $|X|\cdot 2^{2\cdot\ell+\smallo}$.
\end{lemma}

\begin{proof}
The proof of the number of recursive calls is exactly the same as in the proof of Lemma \ref{lem: induction cost}. So we turn to proving the bound on the returned solution's cost.

The proof is by induction on $R+\ell+\smallo$. As in the proof of Lemma \ref{lem: induction cost}, we only need to focus on the case that $\widetilde{\opt}\leq 2\cdot\opt$ and show that $\cost(\mathcal{F}_2)\leq \big(8\cdot(R+ \ell)+1\big)\cdot\opt$.

Let $T=F\cup(\bigcup\limits_{i=1}^R T_i)$ be the partial multi-rooted arborescence obtained from Lemma \ref{lem: di-separator small numb of subinstances, multi-rooted}. Suppose $T$ contains $R'$ many of the roots. Then, exactly $R'$ many of $T_i$'s are non-empty. By Lemma \ref{lem: di-separator small numb of subinstances, multi-rooted} (a) we have that each non-empty $T_i$ is consists of up to four shortest dipaths rooted at $r_i$ so $\cost(T_i)\leq 4\cdot\widetilde{\opt}$ because of the preprocessing step plus the fact that $\widetilde{\opt}\leq 2\cdot\opt$, we conclude
\begin{equation}\label{eq: cost separator, multi-rooted}
    \cost(T\setminus F)\leq 8\cdot R'\cdot\opt.
\end{equation}

Since $T$ contains $R'$ many roots, each subinstance $I_{C_i}$ induced by $(G,T,C_1,\ldots,C_h)$ has at most $R-R'+1$ many roots for $1\leq i\leq h$. Furthermore, by Lemma \ref{lem: di-separator small numb of subinstances, multi-rooted} (b) each $I_{C_i}$'s has at most $\frac{|X|}{2}\leq 2^{\ell-1}$ many terminals. So by induction hypothesis, for $i=1,\ldots,h$ we have
\begin{equation}\label{eq: cost sub instance, multi-rooted}
    \cost(\mathcal{F}_{C_i})\leq \Big(8\cdot\big((R-R'+1)+\ell-1\big)+1\Big)\cdot\opt_{I_{C_i}}\leq \big(8\cdot(R-R'+\ell)+1\big)\cdot\opt_{I_{C_{i}}}.
\end{equation}
Using Lemma \ref{lem: di-separator small numb of subinstances, multi-rooted} (c), the bounds in \eqref{eq: cost separator, multi-rooted} and \eqref{eq: cost sub instance, multi-rooted} we have

\begin{align*}
    \cost(\mathcal{F})&\leq \cost(T\setminus F)+\sum\limits_{i=1}^h\cost(\mathcal{F}_{I_{C_i}})\\
    &\leq 8\cdot R'\cdot\opt+\big(8\cdot(R-R'+\ell)+1\big)\cdot\sum\limits_{i=1}^h\opt_{I_{C_i}}\\
    &\leq 8\cdot R'\cdot\opt+\big(8\cdot(R-R'+\ell)+1\big)\cdot\opt\\
    &= \big(8\cdot(R+\ell)+1\big)\cdot\opt,
\end{align*}
where the third inequality follows from the fact that $\sum\limits_{i=1}^h\opt_{I_{C_i}}\leq\opt$ as the restriction of $\OPT$ on each $G_{C_i}$ is a feasible solution for $I_{C_i}$ and $G_{C_i}$'s are edge-disjoint..
\end{proof}

\begin{pproof}{Theorem \ref{thm: multi-rooted planar DST}}
Note both of the tweaks in Algorithm \ref{alg: log k-approx polytime} are implementable in polynomial time. The proof has exactly the same structure as in the proof of Theorem \ref{thm: planar DST} with the difference that we use Lemma \ref{lem: induction cost, multi-rooted} here instead of Lemma \ref{lem: induction cost}.
\end{pproof}

\section{Concluding Remarks}

One possible direction is to extend our result to minor-free families of graphs. However, as pointed out in \cite{cohen2022bypassing,abraham2006object}, minor-free (undirected) graphs do not have shortest-path separators. In \cite{cohen2022bypassing}, Cohen-Addad bypassed this difficulty by designing a new separator called a {\em mixed separator} for undirected minor-free graphs.
It is not clear that analogous separators exist in directed graphs.
For example, the mixed separators in \cite{cohen2022bypassing} are obtained, in part, by contracting certain paths. These paths are obtained using structural results in minor-free graphs \cite{robertson2003graph} and it is not clear how to find analogous paths in the directed case.
Obtaining an $O(\log k)$-approximation for DST in minor-free graphs remains an interesting open problem.

\bibliographystyle{alpha}
\bibliography{references}

\appendix
\section{Proof of Lemma \ref{lem: poly bounded dist}}\label{app:scale}

\begin{proof}
Let $\Delta:=\max_{t\in X}\Big\{d\big(\{r_1,\ldots ,r_R\},t\big)\Big\}$, i.e., $\Delta$ is the maximum distance from any root to a terminal. Let $\opt_I$ be the optimal value of instance $I$. Then, $\Delta\leq \opt_I\leq k\cdot\Delta$.

If $\Delta=0$, then $\opt_I=0$ and the collection of all shortest dipaths from the roots to the terminals is a solution of cost $0$. So we assume $\Delta>0$.

We can safely remove any edge $e$ having $c_e>k\cdot\Delta$ and any Steiner node $v$ (along with its incident edges) having $d(\{r_1,\ldots ,r_R\},v)>k\cdot\Delta$ since no optimal solution of $I$ uses $e$ or $v$. Since we have only deleted
elements of $G$, it remains planar.

Define a new edge costs $c'_e:=\lceil c_e\cdot\frac{n}{\epsilon\cdot\Delta} \rceil$ and form the instance $I'=(G,c',\{r_1,\ldots,r_R\},X)$.
Note for any shortest dipath $P$ starting at root $r_i$, we have 
\[
\cost_{c'}(P)\leq \sum\limits_{e\in P}c'_e\leq \sum\limits_{e\in P}(c_e\cdot\frac{n}{\epsilon\cdot\Delta}+1)\leq \cost_{c}(P)\cdot\frac{n}{\epsilon\cdot\Delta}+n\leq \frac{n \cdot k}{\epsilon} + n,
\]
where the last inequality follows because all the distances from the root has length at most $k\cdot\Delta$. So all the shortest dipaths starting at $r$ in $I'$ are bounded by $O(\frac{n^2}{\epsilon})$.

Let $\opt_{I'}$ be the optimal value of instance $I'$. Similar calculation as before shows $\opt_{I'}\leq \frac{n}{\epsilon\cdot\Delta}\cdot\opt_I+n$.

Let $F$ be an $\alpha$-approximate solution for $I'$. Then, we have
\begin{align*}
\cost_c(F)&\leq \frac{\epsilon\cdot\Delta}{n}\cdot\cost_{c'}(F)\\
&\leq \frac{\epsilon\cdot\Delta}{n}\cdot\alpha\cdot\opt_{I'}\\
&\leq \frac{\epsilon\cdot\Delta}{n}\cdot\alpha\cdot(\frac{n}{\epsilon\cdot\Delta}\cdot\opt_I+n)\\
&\leq \alpha\cdot\opt_I+\alpha\cdot\epsilon\cdot\Delta\\
&\leq \alpha\cdot(1+\epsilon)\cdot\opt_I,
\end{align*}
where the first inequality follows because $c'_e\geq c_e\cdot\frac{n}{\epsilon\cdot\Delta}$ and the last because $\opt_I \geq \Delta$.
\end{proof}

\end{document}